\pdfoutput=1

\documentclass[conference]{IEEEtran}

\usepackage{amsmath, amssymb,bm}
\usepackage{hyperref}
\usepackage{enumerate,mathtools}
\usepackage{amsthm}
\usepackage{cite}
\usepackage{microtype}
\usepackage{graphicx}

\usepackage{caption}
\usepackage{subcaption}

\usepackage[show]{ed}

\usepackage{tikz}
\usetikzlibrary{arrows}
\usetikzlibrary{graphs}
\usetikzlibrary{shapes.geometric,shapes.arrows,decorations.pathmorphing}
\usetikzlibrary{matrix,chains,scopes,positioning,arrows,fit}
    \usetikzlibrary{trees}
\usepackage{multirow}
\usepackage{bigstrut}
\usepackage{pdfpages}

\usepackage{pgf}
\usetikzlibrary{arrows,automata}
\tikzset{ font={\fontsize{10pt}{12}\selectfont}}

\usepackage{array}
\newcolumntype{L}[1]{>{\raggedright\let\newline\\\arraybackslash\hspace{0pt}}m{#1}}
\newcolumntype{C}[1]{>{\centering\let\newline\\\arraybackslash\hspace{0pt}}m{#1}}
\newcolumntype{R}[1]{>{\raggedleft\let\newline\\\arraybackslash\hspace{0pt}}m{#1}}

\DeclareMathSymbol{\mlq}{\mathord}{operators}{``}
\DeclareMathSymbol{\mrq}{\mathord}{operators}{`'}
\DeclareMathSymbol{\mlqq}{\mathord}{operators}{"5C}
\DeclareMathSymbol{\mrqq}{\mathord}{operators}{`"}

\newtheorem{prop}{Proposition}

\title{ Information-Theoretic Analysis of Refractory\\ Effects in the P300 Speller} 
\author{
\IEEEauthorblockN{Vaishakhi Mayya\IEEEauthorrefmark{1}, Boyla Mainsah\IEEEauthorrefmark{1},  and Galen Reeves\IEEEauthorrefmark{1}\IEEEauthorrefmark{2}}
\IEEEauthorblockA{\IEEEauthorrefmark{1}Department of Electrical and Computer Engineering, Duke University, Durham, NC, USA\\ 
\IEEEauthorrefmark{2}Department of Statistical Science, Duke University, Durham, NC, USA}}

\begin{document}

%
%
%

%
%
%
%
%
%
%
%
%
%
%
%
%
%
%
%
%
%


\maketitle

\begin{abstract}
The P300 speller is a brain-computer interface that enables people with neuromuscular disorders to communicate based on eliciting event-related potentials (ERP) in electroencephalography (EEG) measurements. One challenge to reliable communication is the presence of refractory effects in the P300 ERP that induces temporal dependence in the user's EEG responses. We propose a model for the P300 speller as a communication channel with memory. By studying the maximum information rate on this channel, we gain insight into the fundamental constraints imposed by refractory effects. We construct codebooks based on the optimal input distribution, and compare them to existing codebooks in literature.
\end{abstract}


\section{Introduction}

A brain-computer interface (BCI) is a  system that monitors electrophysiological signals and translates the information encoded in these signals into commands that are relayed to a computer \cite {wolpaw2002}. The P300 speller, developed by Farwell and Donchin \cite{farwell1988},  is a BCI that provides an alternative means of communication for individuals with severe neuromuscular diseases, such as amyotrophic lateral sclerosis \cite{sellers2010}, that impair neural pathways that control muscles. In the extreme case of locked-in syndrome, individuals lose all voluntary muscle control and are unable to communicate verbally or via gestures. 








The P300 speller relies on eliciting and detecting event-related potentials (ERP) embedded in electroencephalography (EEG) data. These ERPs are elicited in response to specific stimulus events within the context of an oddball paradigm. The user is presented with a sequence of stimulus events that fall into one of two classes: a rare oddball, i.e. target stimulus, and a frequent non-target stimulus \cite{sutton1965}. The presentation of the rare target stimulus event elicits an ERP response that includes a distinct positive deflection called the P300 signal.

The P300 allows a user to communicate one character at a time.  In a visual P300 speller, the user is presented with an array of choices on a screen, such as the grid shown in Figure~\ref{fig:screen}. While the user focuses on a target character, subsets of characters, called flash groups, are sequentially illuminated on the grid. In this context, the illumination of a flash group is a stimulus event. Under ideal conditions, a P300 ERP is elicited each time that the target character is presented. The elicited ERPs are embedded in noisy EEG data. Following each stimulus event, a time window of the EEG waveform is analyzed to determine the likelihood that the stimulus event contains the target character. After a sufficient amount of data has been collected, the target character is estimated based on the observed responses. 




                           
\begin{figure}
\centering
  \includegraphics[trim={0 0 0 .5cm},clip, width = 0.5\columnwidth]{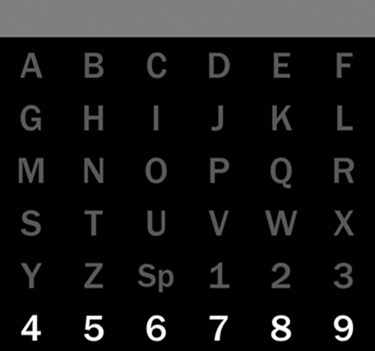}
  \caption{Example of a P300 speller visual interface. The flash group is the set of characters that are illuminated.}
  \label{fig:screen}
\end{figure}

The order in which the target and non-target stimulus events are presented plays a significant role in the ERP elicitation process. This is due, in part, to refractory effects \cite{jin2012}, where the ability to elicit a strong ERP response to every target stimulus event presentation is affected by the target-to-target interval (TTI), which is the amount of time between stimulus events. If a P300 ERP is elicited following the presentation of a target, the amplitude of a successive P300 ERPs elicited in response to subsequent target events with low TTI  may be attenuated or distorted \cite {martens2009}.  The precise behavior of the refractory effects is not well understood and can vary depending on the user and the type of system~\cite{martens2009}.


In this paper, we develop and analyze a communication model for the P300 speller that allows us to understand some of the fundamental constraints imposed by refractory effects. Our channel model consists of a finite state machine (FSM) followed by a memoryless noisy channel; together they form a finite state channel (FSC). The FSM uses $L+1$ states to model a refractory effect that last for $L$ time steps. The memoryless noise channel describes the mapping between an ERP and the response generated by processing the EEG measurements. We study properties of the distributions that maximize the mutual information rate on this channel. We provide an explicit characterization of the optimal input distribution in the noiseless case, and we use the generalized Blahut-Arimoto algorithm (GBAA) \cite{kavcic2001, vontobel2008} to compute the optimal input distributions numerically for the noisy case. We then use the optimal distribution to design codebooks (i.e.\ sequences of flash groups) that can be used for the P300 speller. Performance is assessed using numerical simulations. 
%

The rest of this document is organized as follows. Section~\ref{bgd}, reviews the previous approaches to designing flash groups for the P300 speller and provides relevant background concerning channels with memory. Section~\ref{bci} introduces our channel model and a method to design flash groups using this model. The results in are presented Section~\ref{results}.


\section {Background} \label{bgd}

\subsection{Current codebooks used for P300 speller}

Within the context of the P300 speller, a codebook $\mathbf C$ is a $W \times N$ binary matrix that indicates which of the $W$ character are flashed across $N$ trials. Each row of the matrix corresponds to the flash pattern (or codeword) of a specific character. The columns correspond  to flash groups. Given that a user's target character is $w$, the entry $\mathbf C(w,n)$ indicates whether the target character is flashed in the $n$-th trial.

%
%

Within this setting, previous work have focused on the design of codebooks in order to increase the accuracy of the P300 speller \cite{farwell1988, townsend2010, hill2009, geuze2012, coleman2012, guan04}. In the row-column paradigm (RCP)~\cite{farwell1988},   the codebook is generated using a random permutation of the characters in rows and columns in a grid layout, such as the one shown in Figure \ref{fig:screen}. Due to the randomized order of presentation of the row and column flash groups in the RCP, characters are often flashed twice consecutively.  

The checkerboard paradigm (CBP) \cite{townsend2010} was developed to mitigate refractory effects by imposing a minimum time interval between target character presentations. Due to the method of construction of the codebook in the CBP, the duration of the minimum target interval depends on the specific geometry of the grid layout. 



%
%

Other approaches have used ideas from coding theory to construct codebooks, such as maximizing the minimum Hamming distance~\cite{hill2009,geuze2012}, e.g. the D10 codebook \cite{hill2009}. However, in online studies, these codebooks resulted in similar or worse performance when compared to the RCP. One possible explanation for this behavior is the fact the design of these codebooks did not account for refractory effects. 



The explicit connection between the BCI and an information channel has been considered previously \cite{hill2009, coleman2011}. For example, Omar et al. \cite{coleman2011} represented the BCI-based communication (based on motor imagery)  as a memoryless binary symmetric channel. However, in a P300 speller context, a memoryless channel assumption does not account for system memory due to refractory effects.

\subsection{Information rates for finite-state channels} \label{infchan}

The channel model we study is an indecomposable finite-state channel (IFSC) \cite {Gallager1968}. The channel input is a discrete-time process $\{X_n\}$ supported on a finite alphabet $\mathcal{X}$. The channel state at time $n$ is modeled by a random variable $S_n \in \{1,\cdots, k\}$. The channel output at time $n$ is a random variable $Y_n \in \mathcal{Y}$  whose distribution is a function of the input $X_n$ and state $S_{n-1}$. The channel is defined by the conditional distribution $P(Y_n,S_n | X_n, S_{n-1}) =P(Y_n | X_n, S_{n-1}) P(S_n | X_n, S_{n-1})  $. 

A channel is said to be an intersymbol interference (ISI) channel if the state transitions and the output of the channel depend on current and previous inputs. In this case, it is possible to describe a state sequence that is uniquely determined by the input sequence, i.e.\ $P(Y_n | X_n, S_{n-1}) = P(Y_n | S_n, S_{n-1})$. 

We focus on the setting where the channel input $\{X_n\}$ is an $r$-th order, time-invariant Markov process. Following Kav\v{c}i\'{c}~\cite{kavcic2001}, the distribution on $\{X_n\}$ is parametrized using an $|\mathcal{X}|^r \times |\mathcal{X}|^r$ transition matrix $P$ and the corresponding mutual information rate is defined according to 
\begin{align}
\mathcal R(P) = \lim_{n \to \infty} \frac{1}{n} I(X_1^n ; Y_1^n  | S_0),
\end{align}
where $X_1^n = [X_1, \cdots, X_n]$.  For an IFSC, the limit exists and is independent of the distribution on $S_0$. The maximum information rate over $r$-th order Markov sources is defined according to
\begin{align}
\mathcal R_r = \max_{P \in \mathcal P_r} \mathcal R(P),  \label{eq:R_r}
\end{align}
where $\mathcal P_r$ is the set of all transition matrices for an $r$-th order Markov source. A distribution $P^* \in  \mathcal P_r$ is said to be optimal if it achieves the maximum in \eqref{eq:R_r}. Note that $\mathcal R_r$ provides a lower bound on the capacity of the channel.

Kav\v{c}i\'{c}~\cite{kavcic2001} provides a stochastic method for solving the optimization problem \eqref{eq:R_r} for ISI IFSCs based on a generalization of the Blahut-Arimoto algorithm. This method is extended to a larger class of IFSCs by Vontobel et al.~\cite {vontobel2008}.

\section{Proposed channel model and codebook design}  \label{bci}

\subsection{P300 speller channel model}
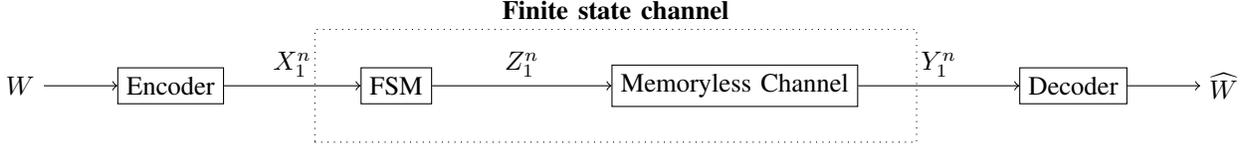
\begin{figure*}[tb]
    \centering
    \tikzstyle{int}=[draw, fill=none]
    \tikzstyle{init} = [pin edge={to-,thin,black}]
        \tikzstyle{line} = [draw, -latex']
    
    \begin{tikzpicture}

        \node  (a) {$W$ };
        
        
        \node (c) [int] [right of=a, node distance=2cm] {Encoder};
        
          \node [int] (e) [right of=c, node distance=3cm] {FSM};
%
           \node [int] (g) [right of=e, node distance=4.5cm] { Memoryless Channel};
%
           \node [int] (i) [right of=g, node distance=4.5cm] {Decoder};
%
          \node (k) [right of=i, node distance=2cm] {$\widehat W$};
        \path[->] (a) edge node{}   (c) ;
          \path[->] (c) edge node[above] {$X^n_1$} (e);
                \path[->] (e) edge node[above] {$Z_1^n$} (g) ;
        \path[->] (g) edge node[above] {$Y_1^n$} (i); 

        
        \draw[->] (i) edge  (k) ;
        
        \node [
    	rectangle,draw,
	dotted,
        	right= 1.2cm of c,
    	minimum width=8cm,
    	minimum height=1.5cm,
    	label=\textbf{Finite state channel},
    ] (pu) {};
    
    \end{tikzpicture}
      \caption{Model of the P300 speller communication channel. The target character, $W$, is encoded with a codeword, $X^n_1$, which is transmitted through a cascade of a finite state machine (FSM), with an intermediate output, $Z_1^n$, and a noisy memoryless channel. The output sequence, $Y^n_1$, is used to obtain an estimate of the target character, $\widehat{W}$. }
    \label{fig:commChannel} 
    \end{figure*}



 The P300 speller is modeled as a cascade of the FSM and a noisy memoryless channel, as shown in Figure~\ref{fig:commChannel}. 
  The states of the FSM model the memory in the channel induced by the refractory period. Throughout this paper, the time step refers to the duration between the presentation of successive flash groups.  A refractory period that lasts $L$ time steps is modeled using $L+1$ states. 

The channel input, $X_n$, is a binary variable that indicates whether the target character is present ($X_n = 1$) or not present ($X_n = 0$) in the flash groups in the $n$-th trial. The state, $S_n$, represents whether the channel is in the ground state, or one of $L$ possible refractory states. The transitions between states are determined according to
\[
S_n = \begin{cases}
G , & \text{if $X_n = 0$, $S_{n-1} = G $ or $S_{n-1} = R_L$}\\
R_l , & \text{if $X_n = 0$, $S_{n-1} = R_{l - 1}, l  \in \{2,\cdots, L\}$ }\\
R_1, &  \text{if $X_n = 1$ },
\end{cases}
\]
and are illustrated in Figure~\ref{fig:fsc}. Note that the state transition is a deterministic function of the previous $L$ channel inputs. 


The intermediate output, $Z_n$, is a binary variable that is equal to one if and only if the input is one and the channel in not in a refractory state, i.e.
\[
Z_n = \begin{cases}
1 , & \text{if $X_n = 1$, $S_{n-1} = G $}\\
0, &  \text{if $X_n = 1$, $S_{n-1} = R_{l}, l  \in \{1, 2 \cdots, L\}$} \\
0, &  \text{if $X_n = 0$ } .
\end{cases}
\]
When the channel is in one of the refractory states, $Z_n = 0$, independently of the input. 
 The output, $Y_n$, is the observed response that depends only on the intermediate output $Z_n$. 

In the context of the P300 speller,  the distribution on $X_1^n$ is induced by the codebook, $\mathbf C$, and the distribution over the target character. The intermediate output $Z_n$ indicates whether a P300 ERP was elicited for the $n$-th trial. The probabilistic  mapping from $Z_n$ to $Y_n$ models the noise induced by the EEG measurement and classification process. 

The case $L = 1$ was studied in our previous work \cite{mayya2016}. In this paper, we study the behavior of channels with general $L$.


 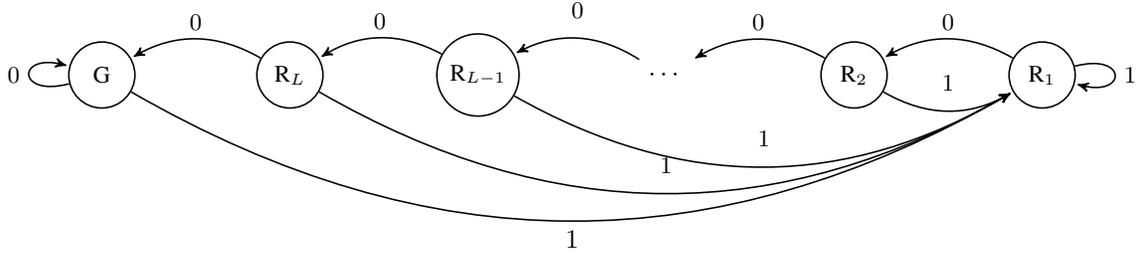
\begin{figure*}[tb]
    \centering
    \tikzstyle{int}=[draw, fill=none]
    \tikzstyle{init} = [pin edge={to-,thin,black}]
        \tikzstyle{line} = [draw, -latex']
    
 \begin{tikzpicture}[->,>=stealth',shorten >=1pt,auto,node distance=2.5cm,
                    semithick]
  \tikzstyle{every state}=[fill=none,draw=black,text=black, font=\small]

  \node[state] (G)                            {$\mbox{G}$};
  \node[state] (RL)    [right of=G]   {$\mbox{R}_{L}$};
  \node[state] (RL1)    [right of=RL] {$\mbox{R}_{L-1}$};
  \node (Rdots)    [right of=RL1]   {$\cdots$};
  \node[state] (R2)    [right of=Rdots] {$\mbox{R}_2$};  
    \node[state] (R1)    [right of=R2] {$\mbox{R}_1$};

  \path    (G) edge [bend right, left,align=center] node [midway,below]{$1$} (R1)
                    edge  [loop left,align= center]      node {\small $0$} (G)
              (RL) edge [bend right, left,align=center]  node [midway,above=4pt]{\small $1$} (R1)
                    edge [bend right, right]                     node [midway,above]{\small$0$} (G)       
              (RL1) edge [bend right, left,align=center]  node [midway,above=4pt]{\small $1$} (R1)
                    edge [bend right, right]                     node [midway,above]{\small$0$} (RL)
              (Rdots) edge [bend right ,align=center]  node [midway,above=4pt]{\small $0$} (RL1)
            (R2) edge [bend right, left,align=center]  node [midway,above=4pt]{\small $1$} (R1)
                    edge [bend right, right]                     node [midway,above]{\small$0$} (Rdots)            
            (R1) edge [loop right,align =  center]     node {\small$1$} (R1)
                    edge [bend right, right]                    node [midway,above]{ \small$0$} (R2);
            
\end{tikzpicture}
      \caption{Model of the P300 ERP elicitation process as an FSM. The input, $X_n$, which governs the state transitions  is marked over the arrows. G is the ground state and $R_l, l = 1,2 \ldots L$ are refractory states. }
    \label{fig:fsc} 
    \end{figure*}

\subsection{Codebook design}\label{sec:codebook} 

Our approach to codebook design is to find distributions on the input sequence $\{X_n\}$ that maximize the mutual information rate in the channel, and then design codebooks that approximate these distributions. For the P300 speller, the channel input is a function of the user's target character and the codebook, $\mathbf C(w,n)$. For a fixed codebook of length $n$, the randomness in the input is due to the randomness in the target character. In order to design codebooks with good properties, we use a random construction in which the rows of the codebook are drawn i.i.d.\ from the distribution that maximizes the mutual information. We refer to this construction as a memory-based codebook for the model with refractory period of length $L$ (MBC($L$)). 


\section{Results} \label{results}

\subsection{Analysis of the optimal input distribution}



This section studies properties of the maximum information rate in the noiseless case, where the observed response $Y_n$ is equal to the intermediate output $Z_n$. In this setting, the optimization problem can be expressed in terms of maximizing entropy rates of run-length limited sequences \cite{zehavi1988,immink1998}.

To facilitate our analysis,  we find it useful to introduce a constrained class of Markov sources. Specifically, we define $\widetilde {\mathcal P}_r $ to be the set of all $r$-th order Markov processes of the form:
\begin{align}
P(X_n = 1|X_{n-r}^{n-1}) = 
\begin{cases}
a, & \text{if }X_{n-r}^{n-1} = \mathbf 0 \\
0, & \text{otherwise}
\end{cases}. \label{eq:P_constrained} 
\end{align}
%
Note that every sequence drawn according to this distribution has at least $r$ zeros between ones. In the rest of the paper, we will focus exclusively on the setting where the memory in source is matched to the memory in the channel, i.e. $r = L$. 




Observe that if the channel input is drawn according to a distribution $P$ in the constrained set $\tilde{\mathcal{P}}_L$, then the mapping between the input and noiseless output is invertible, and thus the mutual information is equal to the entropy of the input:
\begin{align*} 
I(X_1^n;Y_1^n|S_0) &= H(X_1^n|S_0) - H(X_1^n|Y_1^n,S_0)  \\ 
 & = H(X_1^n|S_0). 
\end{align*}
Therefore, the maximum information rate in the constrained setting can be expressed as
\begin{align} 
    \max_{ P \in \mathcal {\widetilde P}_L} \mathcal{R}(P)  &  =  \max_{ P \in \mathcal {\widetilde P}_L}  \lim_{n \rightarrow  \infty} \frac 1 n H(X_1^n|S_0) \notag \\
    & =  \max_{a \in  [0,1]} \frac {H_b(a)}{1 + La}, \label{eqn:Ccon}
\end{align}
where $H_b(p) = - p \log p - (1-p) \log p$ is the binary entropy function. Moreover, by differentiation it can be verified that  the maximum in \eqref{eqn:Ccon} is achieved by the unique solution   $a^* \in  [0,1]$ to the equation 
%
%
\begin{equation} \label{eqn:a}
a = (1-a)^{L + 1}. 
\end{equation} 
In the case of $L = 1$, the optimal value can be computed explicitly as $a^* = \frac {3 - {\sqrt 5}} 2$. As pointed out in \cite{cuff2008}, $1 -a^*$ is the inverse of the golden ratio. 

The next result shows that this rate also provides an upper bound on the mutual information rate. 

\begin{prop} \label{prop1}
Consider the P300 speller channel model in Figure~\ref{fig:commChannel}, with $L$ refractory states. For any distribution on  the channel input $\{X_n\}$ and initial state $S_0$,  the mutual information satisfies 
\begin{align}
\limsup_{n \to \infty} \frac{1}{n} I(X_1^n; Y_1^n  |S_0) \le \mathcal R^\text{UB}_L,
\end{align}
where
\begin{align}
\mathcal R^\text{UB}_L  = \max_{a \in  [0,1] } \frac {H_b(a)}{1 + La}.
\end{align}
\end{prop}
\begin{proof}
Starting with the data processing inequality, we have
\begin{align*} 
\frac{1}{n} I(X_1^n, Y_1^n | S_0) & \le \frac{1}{n} I(X_1^n , Z_1^n | S_0)\\
& =   \frac{1}{n} H( Z_1^n | S_0) ,
\end{align*}
where the second step follows from the chain rule for mutual information and the fact that $Z_n$ is a deterministic function of the channel input. Next, we note that the IFSC must transition through all $L$ refractory states between successive ones, and thus $\{Z_n\}$ is an $(L,\infty)$ constrained binary sequence \cite{zehavi1988,immink1998}. Therefore, by \cite[Theorem~1]{zehavi1988}, it follows that 
\begin{align}
\limsup_{n \to \infty} \frac{1}{n} H(Z_1^n | S_0) \le  \max_{a \in  [0,1] } \frac {H_b(a)}{1 + La}.
\end{align}

\end{proof}

In light of Proposition~\ref{prop1}, we see that as far as the noiseless case is concerned, we can find an optimal distribution by restricting our attention to the constrained set $\widetilde{ \mathcal{P}}_L$. 

\begin{prop}\label{prop2}
Consider the P300 speller channel model in Figure~\ref{fig:commChannel}, with no noise (i.e.\ $Y_n = Z_n)$ and  $L$ refractory states. Let $P^*$ be the optimal distribution in the constrained set $\widetilde{ \mathcal{P}}_L$ with parameter $a^*$ defined by \eqref{eqn:a}. Then, $P^*$ achieves the maximum information rate for the channel, i.e. 
\begin{align}
\mathcal{R}_L(P^*)  =  \max_{ P \in \widetilde{\mathcal{P}}_L} \mathcal{R}_L(P)  = \max_{ P \in P_L} \mathcal{R}_L(P).
\end{align}
%
\end{prop}

We remark that the distribution described in Proposition~\ref{prop2} might not be the only distribution that achieves the maximum.  In \cite{mayya2016}, we showed that in a channel with $L = 1$, when the input is a first order Markov source, there are at least two input transition matrices for which the maximum information rate is achieved. 


In the presence of noise, the problem of optimizing the information rate $\mathcal{R}(P)$ over $r$-th order Markov sources $P \in \mathcal P_r$ can be solved numerically using the GBAA~\cite{vontobel2008}. 



\subsection{Simulation results}

This section uses numerical simulations to compare the performance of our memory-based codebook design with the RCP, CBP, and D10 codebooks. 
We estimate accuracy as a function of the channel noise parameter and the number of states in the channel. 
Following the setup in \cite{mainsah2016}, these simulations apply to the P300 speller layout shown in Figure~\ref{fig:screen}. The noise is modeled using an additive white Gaussian  noise (AWGN) channel with noise power $\sigma^2$. 

Using the GBAA, we find the optimal input transition matrix for general $\sigma^2$. We then generate codebooks that are optimized for channel noise and memory, MBC($L$), as described in Section~\ref{bci}.

In simulations, we select one of 36 characters uniformly as the target character. The codeword associated with that character is transmitted across the channel. The received sequence is used to estimate the target character using an optimal decoder. The accuracy is the percentage of characters that are correctly estimated over 1000 runs. 

Figure~\ref{fig:allresults} compares the performance of the codebooks for channels with $L = 1, 2$ refractory states, where the corresponding MBC($L$) is generated based on the number of channel states. In both cases, the MBC($L$) performs better than all the other codebooks for a given channel.

Figure~\ref{fig:compare3} shows the performance of MBC($L$) associated with a channel with $L = 1, 2, 3$ refractory states. The figure illustrates the decrease in the performance of MBC($L$) as $L$ increases. 

These results illustrate the benefits that can be achieved when the codebooks are designed based on the process underlying the generation of refractory effects. The extent to which our model is representative of refractory effects in the P300 speller is an important direction for future work.  



\begin{figure*}[t!]
\centering
 \begin{subfigure}[t]{0.39\textwidth}
 \includegraphics[width=\textwidth]{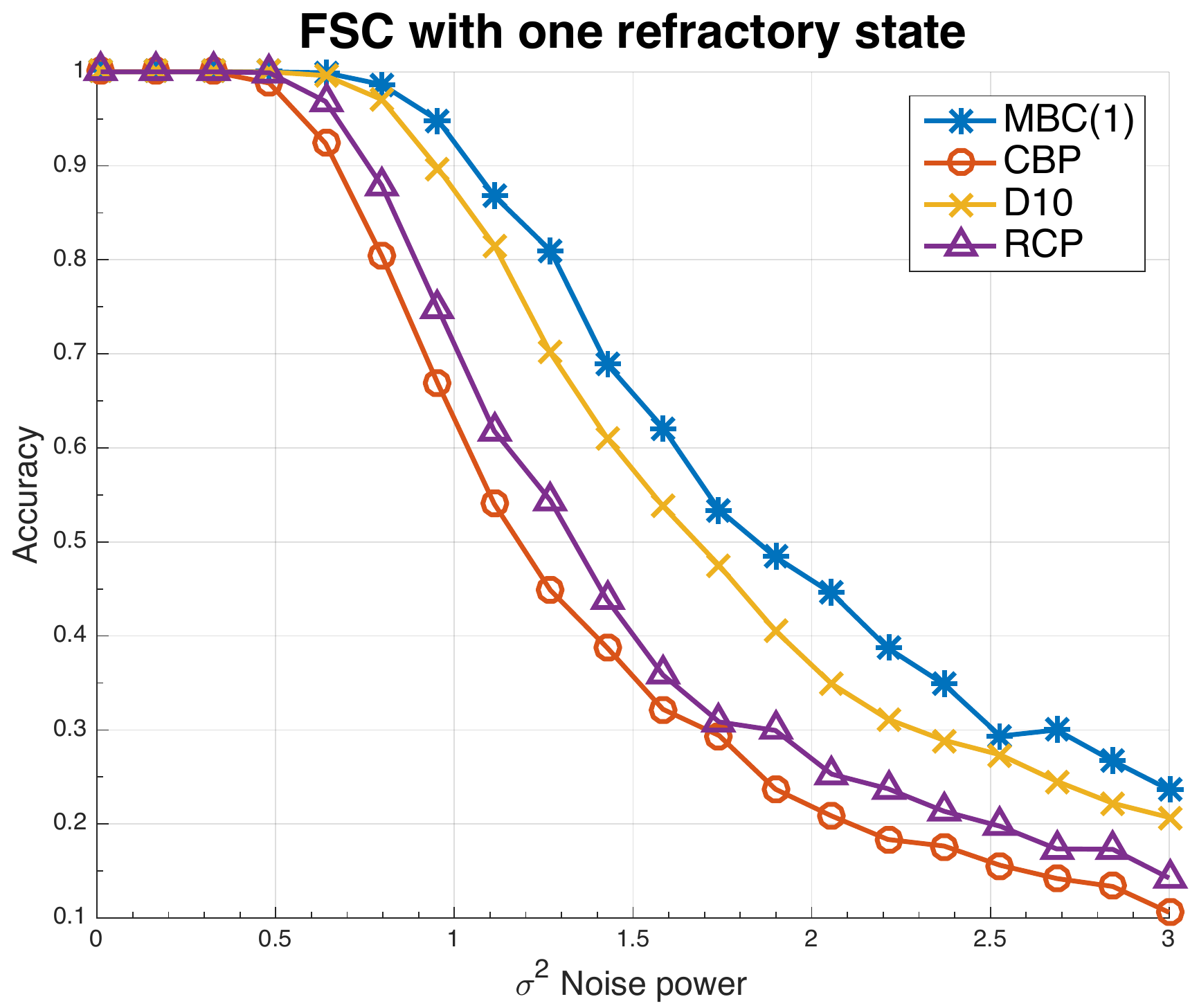}
 \caption{} \label{fig:results1}
\end{subfigure}
 \begin{subfigure}[t]{0.39\textwidth}
 \includegraphics[width=\textwidth]{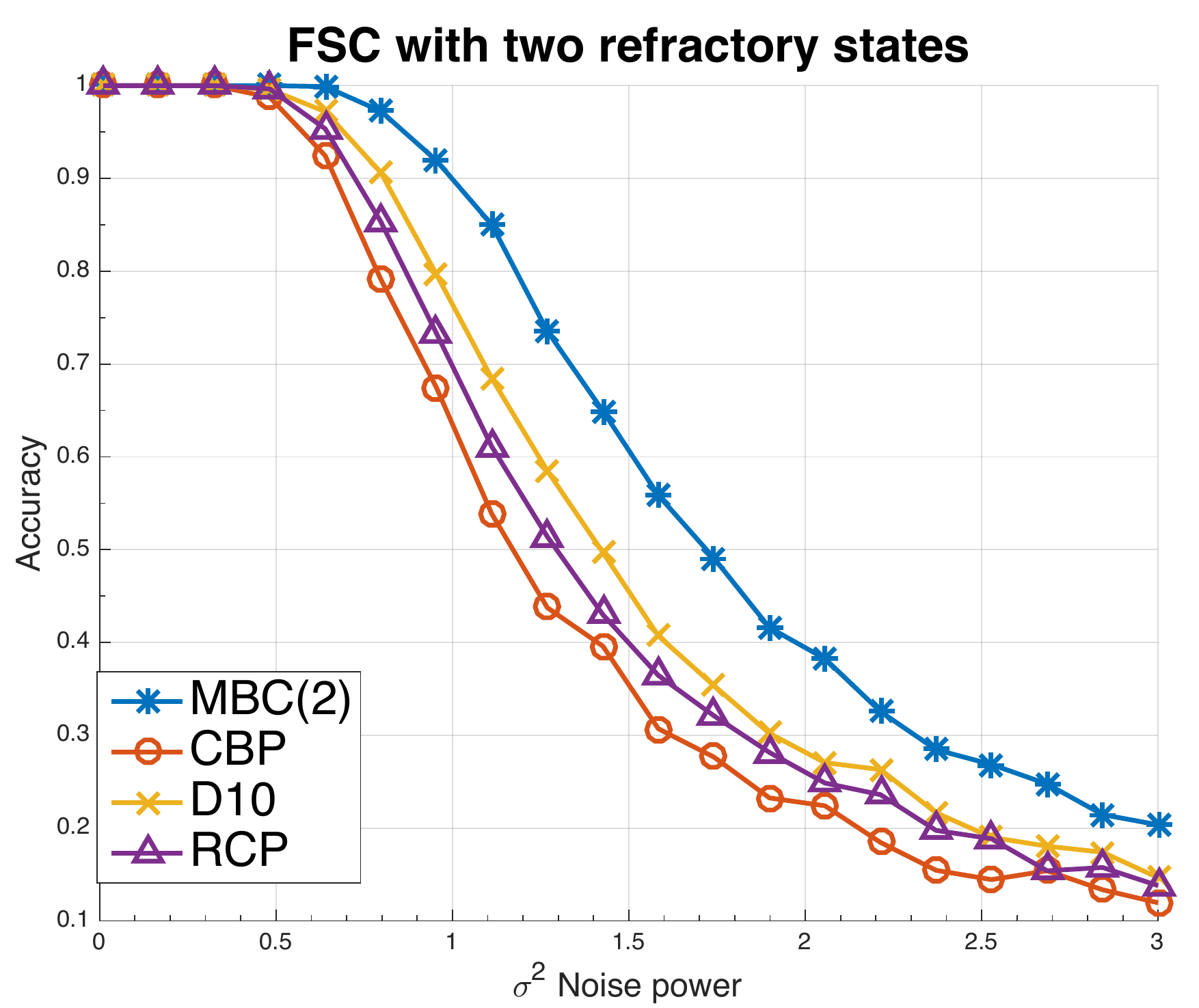}
 \caption{} \label{fig:results2}
\end{subfigure}
\caption{Codebook performance as a function of channel noise parameter, $\sigma^2$, for the P300 speller channel with (a) $L = 1$ and (b) $L = 2$ refractory states. The MBC($L$) outperforms the other codebooks in both channel models. }
\label{fig:allresults}
\end{figure*}

\begin{figure}
\centering
 \includegraphics[width=0.8\columnwidth]{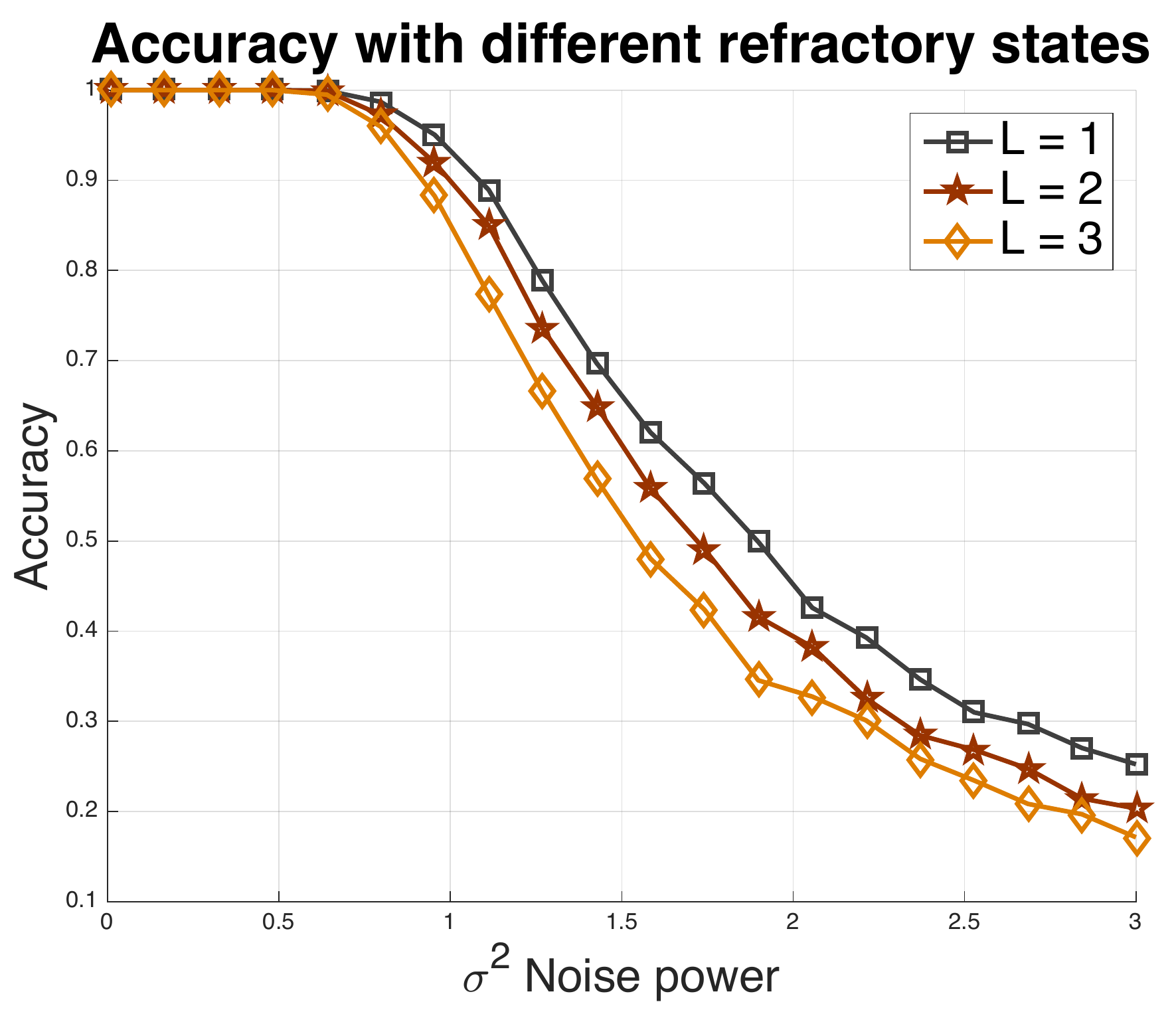}
\caption{Performance of the MBC($L$) as $L$ increases. }
\label{fig:compare3}
\end{figure}

\section{Conclusion} \label{disc}

	
This paper develops a communication model to represent the ERP elicitation process in the P300 speller and then uses this model to design codebooks that are optimized as a function of the length of the refractory period and channel noise. Simulation results suggest that this flexible framework for codebook design could lead to improved performance in settings where refractory effects have a significant impact on the accuracy of the P300 speller. The performance of the memory-based codebook needs to be verified with EEG data and validated with online implementation. 



\bibliography{P300References}
\bibliographystyle{IEEEtran}


%
%
%


\end{document}